\documentclass[a4paper,USenglish,cleveref, autoref, thm-restate]{lipics-v2021}

\title{Finding shortest non-separating and non-disconnecting paths} 

\author{Yasuaki Kobayashi}{Kyoto University, Kyoto, Japan}{kobayashi@iip.ist.i.kyoto-u.ac.jp}
{https://orcid.org/0000-0003-3244-6915}
{JSPS KAKENHI Grant Numbers JP20H05793, JP20K19742, JP21H03499}
\author{Shunsuke Nagano}{Kyoto University, Kyoto, Japan}{shunsuke.mac199921@icloud.com}{}{}

\author{Yota Otachi}
{Nagoya University, Nagoya, Japan
\and \url{https://www.math.mi.i.nagoya-u.ac.jp/~otachi/}}
{otachi@nagoya-u.jp}
{https://orcid.org/0000-0002-0087-853X}
{JSPS KAKENHI Grant Numbers
  JP18H04091, 
  JP18K11168, 
  JP18K11169, 
  JP20H05793, 
  JP21K11752. 
}

\authorrunning{Y. Kobayashi, S. Nagano, Y. Otachi} 

\Copyright{Kobayashi, Nagano, Otachi} 

\ccsdesc[100]{Mathematics of computing~Discrete mathematics~Combinatorics~Combinatorial algorithms} 

\keywords{Non-separating path, Parameterized complexity} 

\relatedversion{} 

\nolinenumbers 

\usepackage{cases}

\newcommand{\dist}{{\rm dist}}
\newcommand{\DP}{{\sf dp}}
\newcommand{\leg}{{\sf leg}}
\newcommand{\rep}{{\sf rep}}
\newcommand{\diam}{{\rm diam}}
\newcommand{\pathseq}[1]{{\langle #1\rangle}}

\EventEditors{John Q. Open and Joan R. Access}
\EventNoEds{2}
\EventLongTitle{42nd Conference on Very Important Topics (CVIT 2016)}
\EventShortTitle{CVIT 2016}
\EventAcronym{CVIT}
\EventYear{2016}
\EventDate{December 24--27, 2016}
\EventLocation{Little Whinging, United Kingdom}
\EventLogo{}
\SeriesVolume{42}
\ArticleNo{23}

\begin{document}

\maketitle

\begin{abstract}
For a connected graph $G = (V, E)$ and $s, t \in V$, a non-separating $s$-$t$ path is a path $P$ between $s$ and $t$ such that the set of vertices of $P$ does not separate $G$, that is, $G - V(P)$ is connected.
An $s$-$t$ path is non-disconnecting if $G - E(P)$ is connected.
The problems of finding shortest non-separating and non-disconnecting paths are both known to be NP-hard.
In this paper, we consider the problems from the viewpoint of parameterized complexity.
We show that the problem of finding a non-separating $s$-$t$ path of length at most $k$ is W[1]-hard parameterized by $k$, while the non-disconnecting counterpart is fixed-parameter tractable parameterized by $k$.
We also consider the shortest non-separating path problem on several classes of graphs and show that this problem is NP-hard even on bipartite graphs, split graphs, and planar graphs.
As for positive results, the shortest non-separating path problem is fixed-parameter tractable parameterized by $k$ on planar graphs and polynomial-time solvable on chordal graphs if $k$ is the shortest path distance between $s$ and $t$.
\end{abstract}

\section{Introduction}

Lov\'asz's path removal conjecture states the following claim:
There is a function $f\colon \mathbb N \to \mathbb N$ such that for every $f(k)$-connected graph $G$ and every pair of vertices $u$ and $v$, $G$ has a path $P$ between $u$ and $v$ such that $G - V(P)$ is $k$-connected.
This claim still remains open, while some spacial cases have been resolved~\cite{ChenGY:Graph:2003,KawarabayashiLRW:weaker:2008,Kriesell:Induced:2001,Tutte:How:1963}.
Tutte~\cite{Tutte:How:1963} proved that $f(1) = 3$, that is, every triconnected graph satisfies that for every pair of vertices, there is a path between them whose removal results a connected graph.
Kawarabayashi et al.~\cite{KawarabayashiLRW:weaker:2008} proved a weaker version of this conjecture:
There is a function $f\colon \mathbb N \to \mathbb N$ such that for every $f(k)$-connected graph $G$ and every pair of vertices $u$ and $v$, $G$ has an induced path $P$ between $u$ and $v$ such that $G - E(P)$ is $k$-connected.

As a practical application, let us consider a network represented by an undirected graph $G$, and we would like to build a private channel between a specific pair of nodes $s$ and $t$.
For some security reasons, the path used in this channel should be dedicated to the pair $s$ and $t$, and hence all other connections do not use all nodes and/or edges on this path while keeping their connections.
In graph-theoretic terms, the vertices (resp.\ edges) of the path between $s$ and $t$ does not form a separator (resp.\ a cut) of $G$.
Tutte's result~\cite{Tutte:How:1963} indicates that such a path always exists in triconnected graphs, but may not exist in biconnected graphs.
In addition to this connectivity constraint, the path between $s$ and $t$ is preferred to be short due to the cost of building a private channel.
Motivated by such a natural application, the following two problems are studied.

\begin{definition}
    Given a connected graph $G$, $s, t \in V(G)$, and an integer $k$, \textsc{Shortest Non-Separating Path} asks whether there is a path $P$ between $s$ and $t$ in $G$ such that the length of $P$ is at most $k$ and $G - V(P)$ is connected.
\end{definition}

\begin{definition}
    Given a connected graph $G$, $s, t \in V(G)$, and an integer $k$, \textsc{Shortest Non-Disconnecting Path} asks whether there is a path $P$ between $s$ and $t$ in $G$ such that the length of $P$ is at most $k$ and $G - E(P)$ is connected.
\end{definition}

We say that a path $P$ is \emph{non-separating} (in $G$) if $G - V(P)$ is connected and is \emph{non-disconnecting} (in $G$) if $G - E(P)$ is connected.

\bigskip
\noindent
\textbf{Related work.}
The shortest path problem in graphs is one of the most fundamental combinatorial optimization problems, which is studied under various settings.
Indeed, our problems \textsc{Shortest Non-Separating Path} and \textsc{Shortest Non-Disconnecting Path} can be seen as variants of this problem.
From the computational complexity viewpoint, \textsc{Shortest Non-Separating Path} is known to be NP-hard and its optimization version cannot be approximated with factor $|V|^{1 - \varepsilon}$ in polynomial time for $\varepsilon > 0$ unless P $=$ NP~\cite{WuC:approximability:2009}.
\textsc{Shortest Non-Disconnecting Path} is shown to be NP-hard on general graphs and polynomial-time solvable on chordal graphs~\cite{Mao:Shortest:2021}.

\bigskip
\noindent
\textbf{Our results.}
We investigate the parameterized complexity of both problems.
We show that \textsc{Shortest Non-Separating Path} is W[1]-hard and \textsc{Shortest Non-Disconnecting Path} is fixed-parameter tractable parameterized by $k$.
A crucial observation for the fixed-parameter tractability of \textsc{Shortest Non-Disconnecting Path} is that the set of edges in a non-disconnecting path can be seen as an independent set of a cographic matroid.
By applying the representative family of matroids~\cite{FominLPS:Efficient:2016}, we show that \textsc{Shortest Non-Disconnecting Path} can be solved in $2^{\omega k}|V|^{O(1)}$ time, where $\omega$ is the exponent of the matrix multiplication.
We also show that \textsc{Shortest Non-Disconnecting Path} is OR-compositional, that is, there is no polynomial kernelization unless coNP $\subseteq$ NP$/$poly.
To cope with the intractability of \textsc{Shortest Non-Separating Path}, we consider the problem on planar graphs and show that it is fixed-parameter tractable parameterized by $k$.
This result can be generalized to wider classes of graphs which have the \emph{diameter-treewidth property}~\cite{Eppstein:Diameter:2000}.
We also consider \textsc{Shortest Non-Separating Path} on several classes of graphs.
We can observe that the complexity of \textsc{Shortest Non-Separating Path} is closely related to that of \textsc{Hamiltonian Cycle} (or \textsc{Hamiltonian Path} with specified end vertices).
This observation readily proves the NP-completeness of \textsc{Shortest Non-Separating Path} on bipartite graphs, split graphs, and planar graphs.
For chordal graphs, we devise a polynomial-time algorithm for \textsc{Shortest Non-Separating Path} for the case where $k$ is the shortest path distance between $s$ and $t$.

\section{Preliminaries}
We use standard terminologies and known results in matroid theory and parameterized complexity theory, which are briefly discussed in this section. See~\cite{Cygan:Parameterized:2015,Oxley:Matroid:2006} for details. 

\bigskip
\noindent
\textbf{Graphs.}
Let $G$ be a graph.
The vertex set and edge set of $G$ are denoted by $V(G)$ and $E(G)$, respectively.
For $v \in V(G)$, the open neighborhood of $v$ in $G$ is denoted by $N_G(v)$ (i.e., $N_G(v) = \{u \in V(G) : \{u, v\} \in E(G)\}$) and the closed neighborhood of $v$ in $G$ is denoted by $N_G[v]$ (i.e., $N_G[v] = N_G(v) \cup \{v\}$).
We extend this notation to sets: $N_G(X) = \bigcup_{v \in X}N_G(v) \setminus X$ and $N_G[X] = N_G(X) \cup X$ for $X \subseteq V(G)$.
For $u, v \in V(G)$, we denote by $\dist_G(u, v)$ the length of a shortest path between $u$ and $v$ in $G$, where the length of a path is defined to the number of edges in it.
We may omit the subscript of $G$ from these notations when no confusion arises.
For $X \subseteq V(G)$, we write $G[X]$ to denote the subgraph of $G$ induced by $X$.
For notational convenience, we may use $G - X$ instead of $G[V(G) \setminus X]$.
For $F \subseteq E$, we also use $G - F$ to represent the subgraph of $G$ consisting all vertices of $G$ and all edges in $E \setminus F$.
For vertices $u$ and $v$, a path between $u$ and $v$ is called a \emph{$u$-$v$ path}.
A vertex is called a \emph{pendant} if its degree is exactly $1$.

\bigskip
\noindent
\textbf{Matroids and representative sets.}
Let $E$ be a finite set.
If $\mathcal I \subseteq 2^E$ satisfies the following axioms, the pair $\mathcal M = (E, \mathcal I)$ is called a \emph{matroid}:
(1) $\emptyset \in \mathcal I$; (2) $Y \in \mathcal I$ implies $X \in \mathcal I$ for $X \subseteq Y \subseteq 2^E$; and (3) for $X, Y \in \mathcal I$ with $|X| < |Y|$, there is $e \in Y \setminus X$ such that $X \cup \{e\} \in \mathcal I$.
Each set in $\mathcal I$ is called an \emph{independent set} of $\mathcal M$.
From the third axiom of matroids, it is easy to observe that every (inclusion-wise) maximal independent set of $\mathcal M$ have the same cardinality.
The \emph{rank} of $\mathcal M$ is the maximum cardinality of an independent set of $\mathcal M$.
A matroid $\mathcal M = (E, \mathcal I)$ of rank $n$ is \emph{linear} (or \emph{representable}) over a field $\mathbb F$ if there is a matrix $M \in \mathbb F^{n \times |E|}$ whose columns are indexed by $E$ such that $X \in \mathcal I$ if and only if the set of columns indexed by $X$ is linearly independent in $M$.

Let $G = (V, E)$ be a graph.
A \emph{cographic matroid} of $G$ is a matroid $\mathcal M(G) = (E, \mathcal I)$ such that $F \subseteq E$ is an independent set of $\mathcal M(G)$ if and only if $G - F$ is connected.
Every cographic matroid is linear and its representation can be computed in polynomial time~\cite{Oxley:Matroid:2006}.

Our algorithmic result for \textsc{Shortest Non-Disconnected Path} is based on \emph{representative families} due to \cite{FominLPS:Efficient:2016}.

\begin{definition}
    Let $\mathcal M = (E, \mathcal I)$ be a matroid and let $\mathcal F \subseteq \mathcal I$ be a family of independent sets of $\mathcal M$.
    For an integer $q \ge 0$, we say that $\widehat{\mathcal F} \subseteq \mathcal F$ is \emph{$q$-representative for $\mathcal F$} if the following condition holds: For every $Y \subseteq E$ of size at most $q$, if there is $X \in \mathcal F$ with $X \cap Y = \emptyset$ such that $X \cup Y \in \mathcal I$, then there is $\widehat{X} \in \widehat{\mathcal F}$ with $\widehat{X} \cap Y = \emptyset$ such that $\widehat{X} \cup Y \in \mathcal I$.
\end{definition}

\begin{theorem}[\cite{FominLPS:Efficient:2016,LokshtanovMPS:Deterministic:2018}]\label{thm:rep}
    Given a linear matroid $\mathcal M = (E, \mathcal I)$ of rank $n$ that is represented as a matrix $M \in \mathbb F^{n \times |E|}$ for some field $\mathbb F$, a family $\mathcal F \subseteq \mathcal I$ of independent sets of size $p$, and an integer $q$ with $p + q \le n$, there is a deterministic algorithm computing a $q$-representative family $\widehat{\mathcal F} \subseteq \mathcal F$ of size $np\binom{p + q}{p}$ with
    \begin{align*}
        O\left(|\mathcal F|\cdot \left(\binom{p + q}{p}p^3n^2 + \binom{p+q}{q}^{\omega - 1} \cdot (pn)^{\omega - 1}\right)\right) + (n + |E|)^{O(1)}
    \end{align*}
    field operations, where $\omega < 2.373$ is the exponent of the matrix multiplication.
\end{theorem}

\bigskip
\noindent
\textbf{Parameterized complexity.}
A \emph{parameterized problem} is a language $L \subseteq \Sigma^* \times \mathbb N$, where $\Sigma$ is a finite alphabet. 
We say that $L$ is \emph{fixed-parameter tractable (parameterized by $k$)} if there is an algorithm deciding if $(I, k) \in L$ for given $(I, k) \in \Sigma^* \times \mathbb N$ in time $f(k)|I|^{O(1)}$, where $f$ is a computable function.
A \emph{kernelization} for $L$ is a polynomial-time algorithm that given an instance $(I, k) \in \Sigma^* \times \mathbb N$, computes an ``equivalent'' instance $(I', k') \in \Sigma^* \times \mathbb N$ such that (1) $(I, k) \in L$ if and only if $(I', k') \in L$ and (2) $|I'| + k' \le g(k)$ for some computable function $g$.
We call $(I', k')$ a \emph{kernel}.
If the function $g$ is a polynomial, then the kernelization algorithm is called a \emph{polynomial kernelization} and its output $(I', k')$ is called a \emph{polynomial kernel}.
An \emph{OR-composition} is an algorithm that given $p$ instances $(I_1, k), \ldots (I_p, k) \in \Sigma^* \times \mathbb N$ of $L$, computes an instance $(I', k') \in \Sigma^* \times \mathbb N$ in time $(\sum_{1 \le i \le p}|I_i| + k)^{O(1)}$ such that (1) $(I', k') \in L$ if and only if $(I_i, k) \in L$ for some $1 \le i \le p$ and (2) $k' = k^{O(1)}$.
For a parameterized problem $L$, its \emph{unparameterized problem} is a language $L' = \{x\#1^k : (x, k) \in L\}$, where$\# \notin \Sigma$ is a blank symbol and $1 \in \Sigma$ is an arbitrary symbol.

\begin{theorem}[\cite{BodlaenderDFH:problems:2009}]\label{thm:composition}
    If a parameterized problem $L$ admits an OR-composition and its unparameterized version is NP-complete, then $L$ does not have a polynomial kernelization unless \emph{coNP} $\subseteq$ \emph{NP}$/$\emph{poly}.
\end{theorem}

\section{\textsc{Shortest Non-Separating Path}}
We discuss our complexity and algorithmic results for \textsc{Shortest Non-Separating Path}.

\subsection{Hardness on graph classes}\label{ssec:graph-classes}

We obverse that, in most cases, \textsc{Shortest Non-Separating Path} is NP-hard on classes of graphs for which \textsc{Hamiltonian Path} (with distinguished end vertices) is NP-hard.
Let $G = (V, E)$ be a graph and $s, t \in V$ be distinct vertices of $G$.
We add a pendant vertex $p$ adjacent to $s$ and denote the resulting graph by $G'$.
Then, we have the following observation.

\begin{observation}\label{lem:complexity:series}
    For every non-separating path $P$ between $s$ and $t$ in $G'$, $V(G) \setminus V(P) = \{p\}$. 
\end{observation}

Suppose that for a class $\mathcal C$ of graphs,
\begin{itemize}
    \item the problem of deciding whether given graph $G \in \mathcal C$ has a Hamiltonian path between specified vertices $s$ and $t$ in $G$ is NP-hard and
    \item $G \in \mathcal C$ implies $G' \in \mathcal C$.
\end{itemize}
By \Cref{lem:complexity:series}, $G'$ has a non-separating $s$-$t$ path if and only if $G$ has a Hamiltonian path between $s$ and $t$.
This implies that the problem of finding a non-separating path between specified vertices is NP-hard on class $\mathcal C$.

\begin{theorem}\label{thm:hardness:classes}
    The problem of deciding if an input graph has a non-separating $s$-$t$ path is NP-complete even on planar graphs, bipartite graphs, and split graphs.
\end{theorem}

The classes of planar graphs and bipartite graphs are closed under the operation of adding a pendant.
Recall that a graph $G$ is a \emph{split graph} if the vertex set $V(G)$ can be partitioned into a clique $C$ and an independent set $I$. Thus, for the class of split graphs, we need the assumption that the pendant added is adjacent to a vertex in $C$.

As the problem trivially belongs to NP, it suffices to show that \textsc{Hamiltonian Path} (with distinguished end vertices) is NP-hard on these classes of graphs\footnote{These results for bipartite graphs and planar graphs seem to be folklore but we were not able to find particular references.}.
For split graphs, it is known that \textsc{Hamiltonian Path} is NP-hard even if the distinguished end vertices are contained 
in the clique $C$~\cite{Muller:Hamiltonian:1996}.
Let $G$ be a graph and let $v \in V(G)$.
We add a vertex $v'$ that is adjacent to every vertex in $N_G(v)$, that is, $v$ and $v'$ are twins.
The resulting graph is denoted by $G'$.
It is easy to verify that $G$ has a Hamiltonian cycle if and only if $G'$ has a Hamiltonian path between $v$ and $v'$.
The class of bipartite graphs is closed under this operation, that is, $G'$ is bipartite.
For planar graphs, $G'$ may not be planar in general.
However, \textsc{Hamiltonian Cycle} is NP-complete even if the input graph is planar and has a vertex of degree $2$~\cite{ItaiPS:Hamilton:1982}.
We apply the above operation to this degree-$2$ vertex, and the resulting graph $G'$ is still planar.
As the problem of finding a Hamiltonian cycle is NP-hard even on bipartite graphs~\cite{Muller:Hamiltonian:1996} and planar graphs~\cite{ItaiPS:Hamilton:1982}, \Cref{thm:hardness:classes} follows.

\subsection{W[1]-hardness}

Next, we show that \textsc{Shortest Non-Separating Path} is W[1]-hard parameterized by $k$.
The proof is done by giving a reduction from \textsc{Multicolored Clique}, which is known to be W[1]-complete~\cite{FellowsHRV:parameterized:2009}.
In \textsc{Multicolored Clique}, we are given a graph $G$ with a partition $\{V_1, V_2, \ldots, V_k\}$ of $V(G)$ and asked to determine whether $G$ has a clique $C$ such that $|V_i \cap C| = 1$ for each $1 \le i \le k$.

From an instance $(G, \{V_1, \ldots, V_k\})$ of \textsc{Multicolored Clique}, we construct an instance of \textsc{Shortest Non-Separating Path} as follows.
Without loss of generality, we assume that $G$ contains more than $k$ vertices.
We add two vertices $s$ and $t$ and edges between $s$ and all $v \in V_1$ and between $t$ and all $v \in V_k$. 
For any pair of $u \in V_i$ and $v \in V_j$ with $|i - j| \ge 2$, we do the following.
If $\{u, v\} \in E$, then we remove it.
Otherwise, we add a path $P_{u,v}$ of length $2$ and a pendant vertex that is adjacent to the internal vertex $w$ of $P_{u, v}$.
Finally, we add a vertex $v^*$, add an edge between $v^*$ and each original vertex $v \in V(G)$, and add a pendant vertex $p$ adjacent to $v^*$.
The constructed graph is denoted by $H$.
See \Cref{fig:W[1]} for an illustration of the graph $H$.

\begin{figure}
    \centering
    \includegraphics[width=0.8\textwidth]{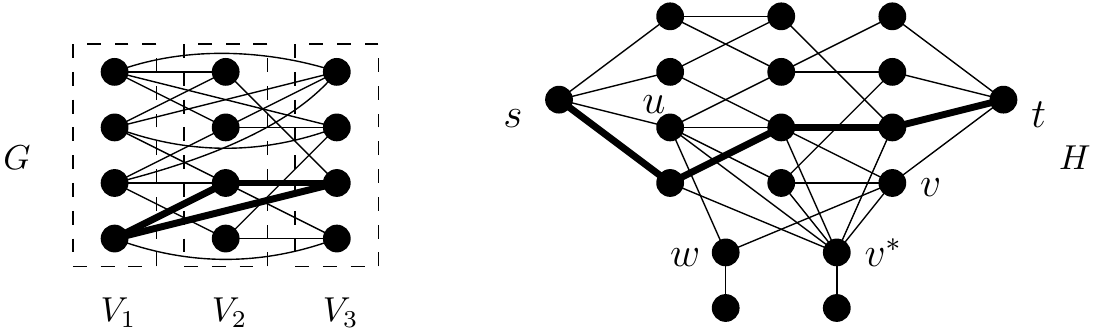}
    \caption{The left figure depicts an instance $G$ of \textsc{Multicolored Clique} and the right figure depicts the graph $H$ constructed from $G$.
    Some vertices and edges in $H$ are not drawn in this figure for visibility. The edges of a clique $C$ and the corresponding non-separating $s$-$t$ path $P$ are drawn as thick lines.}
    \label{fig:W[1]}
\end{figure}

\begin{lemma}\label{lem:complexity:W[1]}
    There is a clique $C$ in $G$ such that $|C \cap V_i| = 1$ for $1 \le i \le k$ if and only if there is a non-separating $s$-$t$ path of length at most $k + 1$ in $H$.
\end{lemma}

\begin{proof}
    Suppose first that $G$ has a clique $C$ with $C \cap V_i = \{v_i\}$ for $1 \le i \le k$.
    Then, $P = \pathseq{s, v_1, v_2, \ldots, v_k, t}$ is an $s$-$t$ path of length $k + 1$ in $H$. 
    To see the connectivity of $H - V(P)$, it suffices to show that every vertex is reachable to $v^*$ in $H - V(P)$.
    By the construction of $H$, each vertex in $V(G) \setminus V(P)$ is adjacent to $v^*$ in $H - V(P)$.
    Each vertex $z$ in $V(H) \setminus (V(G) \cup \{v^*, p\})$ is either the internal vertex $w$ of $P_{u, v}$ for some $u, v \in V(G)$ or the pendant vertex adjacent to $w$.
    In both cases, at least one of $u$ and $v$ is not contained in $P$ as $V(P) \setminus \{s, t\}$ is a clique in $G$, implying that $z$ is reachable to $v^*$.
    
    Conversely, suppose that $H$ has a non-separating $s$-$t$ path $P$ of length at most $k + 1$ in $H$.
    By the assumption that $G$ has more than $k$ vertices, there is a vertex of $G$ that does not belong to $P$.
    Observe that $P$ does not contain any internal vertex $w$ of some $P_{u, v}$ as otherwise the pendant vertex adjacent to $w$ becomes an isolated vertex by deleting $V(P)$, which implies $H - V(P)$ has at least two connected components.
    Similarly, $P$ does not contain $v^*$.
    These facts imply that the internal vertices of $P$ belong to $V(G)$, and we have $|V(P) \cap V_i| = 1$ for $1 \le i \le k$.
    Let $C = V(P) \setminus \{s, t\}$.
    We claim that $C$ is a clique in $G$.
    Suppose otherwise.
    There is a pair of vertices $u, v \in C$ that are not adjacent in $G$.
    This implies that $H$ contains the path $P_{u, v}$.
    However, as $P$ contains both $u$ and $v$, the internal vertex of $P_{u, v}$ together with its pendant vertex forms a component in $H - V(P)$, yielding a contradiction that $P$ is a non-separating path in $H$.
\end{proof}

Thus, we have the following theorem.

\begin{theorem}\label{thm:snsp:w[1]}
    \textsc{Shortest Non-Separating Path} is W[1]-hard parameterized by $k$.
\end{theorem}

\subsection{Graphs with the diameter-treewidth property}

By~\Cref{thm:snsp:w[1]}, \textsc{Shortest Non-Separating Path} is unlikely to be fixed-parameter tractable on general graphs.
To overcome this intractability, we focus on sparse graph classes.
We first note that algorithmic meta-theorems for \textsc{FO Model Checking}~\cite{GroheK:Methods:2009,GroheKS:Deciding:2017} does not seem to be applied to \textsc{Shortest Non-Separating Path} as we need to care about the connectivity of graphs, while it can be expressed by a formula in MSO logic, which is as follows.
The property that vertex set $X$ forms a non-separating $s$-$t$ path can be expressed as:
\begin{align*}
    {\tt conn}(V \setminus X) \land {\tt hampath}(X, s, t),
\end{align*}
where ${\tt conn}(Y)$ and ${\tt hampath}(Y, s, t)$ are formulas in MSO${}_2$ that are true if and only if the subgraph induced by $Y$ is connected and has a Hamiltonian path between $s$ and $t$, respectively.
We omit the details of these formulas, which can be found in \cite{Cygan:Parameterized:2015} for example\footnote{In \cite{Cygan:Parameterized:2015}, they give an MSO${}_2$ sentence ${\tt hamiltonicity}$ expressing the property of having a Hamiltonian cycle, which can be easily transformed into a formula expressing ${\tt hampath}(X, s, t)$.}.
By Courcelle's theorem~\cite{Courcelle:Monadic:1990} and its extension due to Arnborg et al.~\cite{ArnborgLS:Easy:1991}, we can compute a shortest non-separating $s$-$t$ path in $O(f({\rm tw}(G))n)$ time, where $n$ is the number of vertices and ${\rm tw}(G)$ is the treewidth\footnote{We do not give the definition of treewidth and (the optimization version of) Courcelle's theorem. We refer to \cite{Cygan:Parameterized:2015} for details.} of $G$.
As there is an $O(\mathrm{tw}(G)^{\mathrm{tw}(G)^3}n)$-time algorithm for computing the treewidth of an input graph $G$~\cite{Bodlaender:Linear-time:1996}, we have the following theorem.

\begin{theorem}\label{thm:snsp-tw}
    \textsc{Shortest Non-Separating Path} is fixed-parameter tractable parameterized by the treewidth of input graphs.
\end{theorem}

A class $\mathcal C$ of graphs is \emph{minor-closed} if every minor of a graph $G \in \mathcal C$ also belongs to $\mathcal C$.
We say that $\mathcal C$ has the \emph{diameter-treewidth property} if there is a function $f\colon \mathbb N \to \mathbb N$ such that for every $G \in \mathcal C$, the treewidth of $G$ is upper bounded by $f(\diam(G))$, where $\diam(G)$ is the diameter of $G$.
It is well known that every planar graph $G$ has treewidth at most $3\cdot\diam(G) + 1$~\cite{RobertsonS:GM3:1984}\footnote{More precisely, the treewidth of a planar graph is upper bounded by $3r + 1$, where $r$ is the radius of the graph.}, which implies that the class of planar graphs has the diameter-treewidth property.
This can be generalized to more wider classes of graphs.
A graph is called an \emph{apex graph} if it has a vertex such that removing it makes the graph planar.

\begin{theorem}[\cite{DemaineH:Diameter:2004,Eppstein:Diameter:2000}]
    Let $\mathcal C$ be a minor-closed class of graphs.
    Then, $\mathcal C$ has the diameter-treewidth property if and only if it excludes some apex graph.
\end{theorem}

For $C \subseteq V(G)$ that induces a connected subgraph $G[C]$, we denote by $G_C$ the graph obtained from $G$ by contracting all edges in $G[C]$ and by $v_C$ the vertex corresponding to $C$ in $G_C$.
Since $G[C]$ is connected, vertex $v_C$ is well-defined.

\begin{lemma}\label{lem:snsp:comp}
    Let $C \subseteq V(G)$ be a vertex subset such that $G[C]$ is connected.
    Let $P$ be an $s$-$t$ path in $G$ with $V(P) \cap C = \emptyset$.
    Then, $P$ is non-separating in $G$ if and only if it is non-separating in $G_C$.
\end{lemma}

\begin{proof}
    Suppose first that $P$ is non-separating in $G$.
    Let $u, v \in V(G) \setminus V(P)$ be arbitrary.
    As $P$ is non-separating, there is a $u$-$v$ path $P'$ in $G - V(P)$.
    Let $u'$ be the vertex of $G_C$ such that $u' = u$ if $u \notin C$ and $u' = v_C$ if $u \in C$.
    Let $v'$ be the vertex defined analogously. 
    We show that there is a $u'$-$v'$ path in $G_C - V(P)$ as well.
    If $P'$ does not contain any vertex in $C$, it is also a $u'$-$v'$ path in $G_C$, and hence we are done.
    Suppose otherwise.
    Let $x$ and $y$ be the vertices in $V(P') \cap C$ that are closest to $u$ and $v$, respectively.
    Note that $x$ and $y$ can be the end vertices of $P'$, that is, $C$ may contain $u$ and $v$.
    Let $P_{u,x}$ and (resp.\ $P_{y,v}$) be the subpath of $P'$ between $u$ and $x$ (resp.\ $y$ and $v$).
    Then, the sequence of vertices obtained by concatenating $P_{u, x}$ after $P_{y, v} - \{y\}$ and replacing exactly one occurrence of a vertex in $C$ with $v_C$ forms a path between $u'$ and $v'$.
    Since we choose $u, v$ arbitrarily, there is a path between any pair of vertices in $G_C - V(P)$ as well. 
    Hence, $P$ is non-separating in $G_C$.
    
    Conversely, suppose that $P$ is non-separating in $G_C$.
    For $u, v \in V(G_C) \setminus V(P)$, there is a path $P'$ in $G_C - V(P)$.
    Suppose that neither $u = v_C$ nor $v = v_C$.
    Then, we can construct a $u$-$v$ path in $G - V(P)$ as follows.
    If $v_C \notin V(P')$, $P'$ is also a path in $G - V(P)$ and hence we are done.
    Otherwise, $v_C \in V(P')$.
    Let $P_u$ and $P_v$ be the subpaths in $P' - \{v_C\}$ containing $u$ and $v$, respectively.
    From $P_u$ and $P_v$, we have a $u$-$v$ path in $G$ by connecting them with an arbitrary path in $G[C]$ between the end vertices other than $u$ and $v$.
    Note that such a bridging path in $G[C]$ always exists since $G[C]$ is connected.
    Moreover, as $V(P') \cap C = \emptyset$ and $V(P) \cap C = \emptyset$, this is also a $u$-$v$ path in $G - V(P)$.
    Suppose otherwise that either $u = v_C$ or $v = v_C$, say $u = v_C$.
    In this case, we can construct a path between every vertex $w$ in $C$ and $v$ by concatenating $P'$ and an arbitrary path in $G[C]$ between $w$ and the end vertex of the subpath $P' - \{v_C\}$ other than $v$.
    Therefore, $P$ is non-separating in $G$.
\end{proof}

Now, we are ready to state the main result of this subsection.

\begin{theorem}\label{thm:snsp-dt}
    Suppose that a minor-closed class $\mathcal C$ of graphs has the diameter-treewidth property.
    Then, \textsc{Shortest Non-Separating Path} on $\mathcal C$ is fixed-parameter tractable parameterized by $k$.
\end{theorem}

\begin{proof}
    Let $G \in \mathcal C$.
    We first compute $B = \{v \in V(G) : \dist(s, v) \le k\}$.
    This can be done in linear time.
    If $t \notin B$, then the instance $(G, s, t, k)$ is trivially infeasible.
    Suppose otherwise that $t \in B$.
    Let $C$ be a component in $G - B$.
    By definition, every non-separating $s$-$t$ path $P$ of length at most $k$ does not contain any vertex of $C$.
    Let $G'$ be the graph obtained from $G$ by contracting all edges in $E(G - B)$.
    For each component $C$ in $G - B$, we denote by $v_C$ the vertex of $G'$ corresponding to $C$ (i.e., $v_C$ is the vertex obtained by contracting all edges in $C$).
    Then, we have $\diam(G') \le 2k + 2$ as $\diam(G[B]) \le k$ and every vertex in $V(G') \setminus B$ is adjacent to a vertex in $B$.
    By~\Cref{lem:snsp:comp}, $G$ has a non-separating $s$-$t$ path of length at most $k$ if and only if so does $G'$.
    Since $\mathcal C$ is minor-closed, we have $G' \in \mathcal C$ and hence the treewidth of $G'$ is upper bounded by $f(2k + 2)$ for some function $f$.
    By~\Cref{thm:snsp-tw}, we can check whether $G'$ has a non-separating $s$-$t$ path of length at most $k$ in $O(g(k)|V(G')|)$ time for some function $g$.
\end{proof}

\subsection{Chordal graphs with $k = \dist(s, t)$}

In \Cref{ssec:graph-classes}, we have seen that \textsc{Shortest Non-Separating Path} is NP-complete even on split graphs (and thus on more general chordal graphs as well).
To overcome this intractability, we restrict ourselves to finding a non-separating $s$-$t$ path of length $\dist(s, t)$ on chordal graphs.

A graph $G$ is \emph{choral} if it has no cycles of length at least~$4$ as an induced subgraph. 
In the following, fix a connected chordal graph $G$.

\begin{lemma}
\label{lem:induced-path-in-connected-sets}
Let $S \subseteq V(G)$ be a vertex set such that $G[S]$ is connected.
For $u, v \in S$, every induced $u$-$v$ path $P$ in $G$ satisfies that $V(P) \subseteq N[S]$.
\end{lemma}
\begin{proof}
Suppose to the contrary that an induced $u$-$v$ path $P$ contains a vertex $x \notin N[S]$.
Since $P$ starts and ends in $S$,
it contains a subpath $Q = \pathseq{a, \dots, x, \dots, b}$ such that
$a, b \in N(S)$ and all other vertices in $Q$ belong to $V - N[S]$.
As $a,b \in N(S)$ and $G[S]$ is connected,
$G$ has an induced $a$-$b$ path $R$ with all internal vertices belonging to $S$.
Since the internal vertices of $Q$ have no neighbors in $S$,
$Q \cup R$ induces a cycle.
Both $a$-$b$ paths $Q$ and $R$ have length at least $2$ as $a$ and $b$ are not adjacent,
and thus the cycle $G[Q \cup R]$ has length at least~$4$.
This contradicts that $G$ is chordal.
\end{proof}

For $u, v \in V(G)$, a set of vertices $S \subseteq V(G) \setminus \{u, v\}$ is called a \emph{$u$-$v$ separator} of $G$ if there is no $u$-$v$ path in $G - S$.
An inclusion-wise minimal $u$-$v$ separator of $G$ is called a \emph{minimal $u$-$v$ separator}.
A \emph{minimal separator} of $G$ is a minimal $u$-$v$ separator for some $u, v \in G$.
Dirac's well-know characterization~\cite{Dirac:rigid:1961} of chordal graphs states that a graph is chordal if and only of every minimal separator induces a clique.

\begin{lemma}
\label{lem:shortest-path-sep}
Let $s,t \in V(G)$ be such that  $\{s,t\} \notin E(G)$.
If $v \in V(G) \setminus \{s,t\}$ is an internal vertex of a shortest $s$-$t$ path $P$,
then $N[v] \setminus \{s,t\}$ is an $s$-$t$ separator of $G$.
\end{lemma}
\begin{proof}
Let $d = \dist(s,t)$.
For $0 \le i \le d$, let 
\begin{align*}
    D_{i} = \{v \in V(G) : \dist(s,v) = i \land \dist(v,t) = d-i\}
\end{align*}
and $V(P) \cap D_{i} = \{u_{i}\}$. Let $j$ ($0 < j < d$) be the index such that $v = u_{j}$.

Suppose to the contrary that there is an induced $s$-$t$ path $Q$
such that $V(Q) \cap (N[u_{j}] \setminus \{s,t\}) = \emptyset$.
By \Cref{lem:induced-path-in-connected-sets},
$V(Q) \subseteq N[V(P)] = \bigcup_{0 \le i \le d} N[u_{i}]$ holds.
Since $Q$ starts in $N[u_{0}]$ and ends in $N[u_{d}]$,
there are indices $i$ and $k$ with $0 \le i < j < k \le d$ such that
$Q$ consecutively visits a vertex $v_{i} \in N[u_{i}]$ 
and then a vertex $v_{k} \in N[u_{k}]$ in this order.
Since $\dist(u_{i}, u_{k}) = k-i \ge 2$ and $\{v_{i}, v_{k}\} \in E$,
at least one of $v_{i} \ne u_{i}$ and $v_{k} \ne u_{k}$ holds.
By symmetry, we assume that $v_{i} \ne u_{i}$.

If $v_{k} = u_{k}$, then $v_{i} \in N(u_{i}) \cap N(u_{k})$.
In this case, we have $i=j-1$ and $k=j+1$
since otherwise $P$ admits a shortcut using the subpath $\pathseq{u_{i}, v_{i}, u_{k}}$.
This implies that 
$\dist(s, v_{i}) \le \dist(s, u_{i}) + 1 = i+1 = j$ and 
$\dist(v_{i}, t) \le 1 + \dist(v_{k}, t) = 1 + \dist(u_{k}, t) = 1+(d-k) = d-j$.
Since $\dist(s, v_{i}) + \dist(v_{i}, t) \ge d$,
we have $\dist(s, v_{i}) = j$ and $\dist(v_{i}, t) = d-j$.
This implies that $v_{i} \in D_{j} \subseteq N[u_{j}] \setminus \{s,t\}$, a contradiction

Next we consider the case $v_{k} \ne u_{k}$. Recall that we also have $v_{i} \ne u_{i}$ as an assumption.
In this case, we have $k-i \le 3$ as $\pathseq{u_{i}, v_{i}, v_{k}, u_{k}}$ is not a shortcut for $P$.
Assume first that $k-i = 3$. By symmetry, we may assume that $i = j-1$ and $k = j+2$.
Since $\dist(s,v_{i}) \le \dist(s,u_{i}) +1 = j$ and 
$\dist(v_{i},t) \le 2 + \dist(u_{k},t) \le 2 + (d-k) = d-j$,
we have $v_{i} \in D_{j} \subseteq N[u_{j}] \setminus \{s,t\}$, a contradiction.
Next assume that $k-i=2$. That is, $i = j-1$ and $k = i+1$.
Since $v_{i}, v_{k} \notin N[u_{j}] \setminus \{s,t\}$
and $P$ is shortest,
the vertices $v_{i}$, $u_{i}$, $u_{j}$, $u_{k}$, $v_{k}$
are distinct and form a cycle of length $5$.
Observe that $v_{i} \notin \{s,t\}$ since otherwise 
$\pathseq{v_{i} = s, v_{k}, u_{k}}$ or $\pathseq{u_{i}, v_{i} = t}$ is a shortcut.
Similarly, $v_{k} \notin \{s,t\}$.
Hence, $v_{i}, v_{k} \notin N[u_{j}]$.
Therefore, the possible chords for the cycle $\pathseq{v_{i}, u_{i}, u_{j}, u_{k}, v_{k}}$
 are $\{u_{i}, v_{k}\}$ and $\{u_{k}, v_{i}\}$.
In any combination of them, the graph has an induced cycle of length at least $4$.
\end{proof}

Let $d$ and $D_i$ be defined as in the proof of \Cref{lem:shortest-path-sep}, and let $D = \bigcup_{0 \le i \le d} D_{i}$.
Note that each $D_{i}$ is a clique: if $i \in \{0,d\}$, then it is a singleton;
otherwise, it is a minimal $s$-$t$ separator of the chordal graph $G[D]$.
Observe that if $|D_{i}| = 1$ for all $0 \le i \le d$,
then $G$ contains a unique shortest $s$-$t$ path, and thus the problem is trivial.
Otherwise, we define $\ell$ to be the minimum index such that $|D_{\ell}| > 1$ and
$r$ to be the maximum index such that $|D_{r}| > 1$.
Since $|D_{0}| = |D_{d}| = 1$, we have $0 < \ell \le r < d$.

Our algorithm works as follows.
\begin{enumerate}
  \item If $G$ contains a unique shortest $s$-$t$ path $P$,
  then test if $P$ is non-separating.
  
  \item Otherwise, find a shortest $s$-$t$ path $P$ satisfying the following conditions.
  \begin{enumerate}
    \item $V(P)$ does not contain a minimal $a$-$b$ separator for $a \in D_{\ell}$ and $b \in V \setminus D$.
    \label{itm:D-W}

    \item $V(P)$ does not contain a minimal $a$-$b$ separator for $a \in D_{\ell}$ and $b \in D_{r}$.
    \label{itm:ell-r}
  \end{enumerate}
\end{enumerate}

\begin{lemma}
The algorithm is correct.
\end{lemma}
\begin{proof}
The first case is trivial.
In the following, we prove the correctness of the second case.

First we show that the condition \ref{itm:D-W} is necessary.
Let $a \in D_{\ell}$ and $b \in V \setminus D$.
Since $|D_{\ell}| > 1$ and $|V(P) \cap D_{\ell}| = 1$,
there is a vertex $a' \in D_{\ell} \setminus V(P)$,
where $a'$ may be $a$ itself.
Since $V(P) \subseteq D$, it holds that $b \notin V(P)$.
Hence, $a'$ and $b$ belong to the same connected component of $G - V(P)$.
Since $D_{\ell}$ is a clique, $a \in N[a']$.
Thus, $a$ and $b$ belong to the same connected component of $G-(V(P) \setminus \{a,b\})$.
Therefore, $V(P)$ does not contain any $a$-$b$ separator.

Next we show that the condition \ref{itm:ell-r} is necessary.
Let $a \in D_{\ell}$ and $b \in D_{r}$.
As before, it suffices to show that
$a$ and $b$ belong to the same connected component of $G-(V(P) \setminus \{a,b\})$.
By the same reasoning in the previous case,
there are vertices $a' \in D_{\ell} \setminus V(P)$ and $b' \in D_{r} \setminus V(P)$
and they belong to the same connected component of $G-V(P)$.
Now, since $a \in N[a']$ and $b \in N[b']$,
$a$ and $b$ belong to the same connected component of $G-(V(P) \setminus \{a,b\})$.

Finally we show that the conditions \ref{itm:D-W} and \ref{itm:ell-r}
together form a sufficient condition for $P$ to be non-separating.
Assume that a shortest $s$-$t$ path $P$ satisfies the conditions \ref{itm:D-W} and \ref{itm:ell-r}.
Since $|D_{\ell}| > 1$ and $|V(P) \cap D_{\ell}| = 1$,
there is a connected component $C$ of $G - V(P)$ that contains at least one vertex of $D_{\ell}$.
Now the condition \ref{itm:D-W} implies that $V \setminus D \subseteq V(C)$ (recall that $V(P) \subseteq D$),
and the condition \ref{itm:ell-r} implies that $(D_{\ell} \cup D_{r}) \setminus V(P) \subseteq V(C)$ holds.
To complete the proof, it suffices to show that $D_{i} \setminus V(P) \subseteq V(C)$ for all $i$.
If $i < \ell$ or $i > r$, then $D_{i} \setminus V(P) = \emptyset$.
Let $v \in D_{i} \setminus V(P)$ for some $i$ with $\ell \le i \le r$.
Observe that $v$ is an internal vertex of a shortest path
from the unique vertex $u \in D_{\ell-1}$
to the unique vertex $w \in D_{r+1}$.
By \Cref{lem:shortest-path-sep}, $N[v] \setminus \{u,w\}$ is a $u$-$w$ separator.
Since $C$ is connected and $u, w$ have neighbors in $C$, 
$G[V(C) \cup \{u,w\}]$ contains a $u$-$w$ path $Q$.
Since $N[v] \setminus \{u,w\}$ is a $u$-$w$ separator,
$Q$ contains a vertex $q$ such that
\[
  q \in V(Q) \cap (N[v] \setminus \{u,w\}) = (V(Q)\setminus \{u,w\}) \cap N[v] \subseteq V(C) \cap N[v].
\]
Therefore, $v$ has a neighbor (i.e., $q$) in $V(C)$, and thus $v$ itself belongs to $C$.
\end{proof}

\begin{lemma}
The algorithm has a polynomial-time implementation.
\end{lemma}
\begin{proof}
Since $G$ is chordal, each minimal separator of $G$ is a clique.
Since $P$ is a shortest path, the size of a clique in $G[V(P)]$ is at most $2$.
Therefore, every minimal separator of $G$ contained in $V(P)$ has size at most $2$.
Furthermore, every size-$2$ minimal separator $\{u,v\}$ is an edge of $G$.
This observation gives us the following implementation of the algorithm
that clearly runs in polynomial time.

For $i \in \{1,2\}$, let $\mathcal{F}_{i}$ be the set of size-$i$ minimal $a$-$b$ separators of $G$
such that $a \in D_{\ell}$ and $b \in (V \setminus D) \cup D_{r}$.
It suffices to find a shortest $s$-$t$ path $P$ such that 
no element of $\mathcal{F}_{1} \cup \mathcal{F}_{2}$ is a subset of $V(P)$.
To forbid the elements of $\mathcal{F}_{1}$, we just remove the vertices that form the size-$1$ separators in $\mathcal{F}_{1}$.
Similarly, to forbid the elements of $\mathcal{F}_{2}$, we remove the edges corresponding to the size-$2$ separators in $\mathcal{F}_{2}$.
Now we find a shortest $s$-$t$ path $P$ in the resultant graph.
If $P$ has length $d = \dist_{G}(s,t)$, then $P$ is a non-separating shortest $s$-$t$ path in $G$.
Otherwise, $G$ does not have such a path.
\end{proof}

\begin{theorem}\label{thm:choral:poly}
    There is a polynomial-time algorithm for \textsc{Shortest Non-Separating Path} on chordal graphs, provided that $k$ is equal to the shortest path distance between $s$ and $t$.
\end{theorem}

\section{\textsc{Shortest Non-Disconnecting Path}}

The goal of this section is to establish the fixed-parameter tractability and a conditional lower bound on polynomial kernelizations for \textsc{Shortest Non-Disconnecting Path}.

\subsection{Fixed-parameter tractability}

\begin{theorem}\label{thm:sndp}
    \textsc{Shortest Non-Disconnecting Path} can be solved in time $2^{\omega k}n^{O(1)}$, where $\omega$ is the matrix multiplication exponent and $n$ is the number of vertices of the input graph $G$.
\end{theorem}

To prove this theorem, we give a dynamic programming algorithm with the aid of representative families of cographic matroids.
Let $(G, s, t, k)$ be an instance of \textsc{Shortest Non-Disconnecting Path}.
For $0 \le i \le k$ and $v \in V(G)$, we define $\DP(i, v)$ as the family of all sets of edges $F$ satisfying the following two conditions: (1) $F$ is the set of edges of an $s$-$v$ path of length $i$ and (2) $G - F$ is connected.
An edge set $F$ is \emph{legitimate} if $F$ forms a path and $G - F$ is connected.
For a family of edge sets $\mathcal F$ and an edge $e$, we define $\mathcal F \Join e := \{F \cup \{e\} : F \in \mathcal F\}$ and $\leg(\mathcal F)$ as the subfamily of $\mathcal F$ consisting of all legitimate $F \in \mathcal F$.
The following simple recurrence correctly computes $\DP(i , v)$.

\begin{numcases}{\DP(i, v) = }
    \{\emptyset\} & $i = 0$ and $s = v$ \label{eq:1}\\
    \emptyset & $i = 0$ and $s \neq v$ \label{eq:2}\\
    \leg\left(\displaystyle\bigcup_{u \in N(v)} (\DP(i-1, u) \Join \{u, v\})\right) & $i > 0$ \label{eq:3}.
\end{numcases}

A straightforward induction proves that $\DP(i, t) \neq \emptyset$ if and only if $G$ has a non-disconnecting $s$-$t$ path of length exactly $i$ and hence it suffices to check whether $\DP(i, t) \neq \emptyset$ for $0 \le i \le k$.
However, the running time to evaluate this recurrence is $n^{O(k)}$.
To reduce the running time of this algorithm, we apply \Cref{thm:rep} to each $\DP(i, v)$.
Now, instead of (\ref{eq:3}), we define
\begin{align}
    \DP(i, v) = \rep_{k - i}\left(\leg\left(\displaystyle\bigcup_{u \in N(v)} (\DP(i-1, u) \Join \{u, v\})\right)\right), \tag{3'}\label{eq:4}
\end{align}
where $\rep_{k - i}(\mathcal F)$ is a $(k-i)$-representative family of $\mathcal F$ for the cographic matroid $\mathcal M = (E(G), \mathcal I)$ defined on $G$.
In the following, we abuse the notation of $\DP$ to denote the families of legitimate sets that are computed by the recurrence composed of (\ref{eq:1}), (\ref{eq:2}), and (\ref{eq:4}).

\begin{lemma}\label{lem:dp-correctness}
    The recurrence composed of (\ref{eq:1}), (\ref{eq:2}), and (\ref{eq:4}) is correct, that is, $G$ has a non-disconnecting $s$-$t$ path of length at most $k$ if and only if $\bigcup_{0 \le i \le k}\DP(i, t) \neq \emptyset$.
\end{lemma}

\begin{proof}
    It suffices to show that $\DP(k', t) \neq \emptyset$ if $G$ has a non-disconnecting $s$-$t$ path $P$ of length $k' \le k$.
    Let $P = (v_0 = s, v_1, \ldots, v_{k'} = t)$ be a non-disconnecting path in $G$.
    For $0 \le i \le k'$, we let $P_i = (v_i, v_{i + 1}, \ldots, v_k)$.
    In the following, we prove, by induction on $i$, a slightly stronger claim that there is a legitimate set $F \in \DP(i, v_i)$ such that $F \cup E(P_i)$ forms a non-disconnecting $s$-$t$ path in $G$ for all $0 \le i \le k'$.
    As $\DP(0, s) = \{\emptyset\}$ and $P_0 = P$ itself is a non-disconnecting path, we are done for $i = 0$.
    Suppose that $i > 0$.
    By the induction hypothesis, there is a legitimate $F \in \DP(i - 1, v_{i - 1})$ such that $F \cup E(P_{i-1})$ forms a non-disconnecting $s$-$t$ path in $G$.
    Let $\mathcal F = \leg(\bigcup_{u \in N(v_i)}(\DP(i - 1) \Join \{u, v_i\}))$.
    Since $F \cup E(P_{i - 1})$ is legitimate, $F \cup \{\{v_{i - 1}, v_i\}\}$ is also legitimate, implying that $\mathcal F$ is nonempty.
    Let $\widehat{\mathcal F} = \rep_{k - i}(\mathcal F)$ be $(k-i)$-representative for $\mathcal F$ and let $Y = \{\{v_j, v_{j + 1}\} : i \le j < k'\}$.
    As $|Y| \le k - i$, $X \cap Y = \emptyset$, and $X \cup Y \in \mathcal I$, $\widehat{\mathcal F}$ contains an edge set $\widehat{X}$ with $\widehat{X} \cap Y$ and $\widehat{X} \cup Y \in \mathcal I$, implying that there is $\widehat{X} \in \DP(i, v_i)$ such that $\widehat{X}\cup E(P_i)$ forms a non-disconnecting $s$-$t$ path in $G$.
    Thus, the lemma follows.
\end{proof}

\begin{lemma}\label{lem:dp-time}
    The recurrence can be evaluated in time $2^{\omega k}n^{O(1)} \subset 5.18^kn^{O(1)}$, where $\omega < 2.373$ is the exponent of the matrix multiplication.
\end{lemma}

\begin{proof}
    By~\Cref{thm:rep}, $\DP(i, v)$ contains at most $2^kkn$ sets for $0 \le i \le k$ and $v \in V(G)$ and can be computed in time $2^{\omega k}n^{O(1)}$ by dynamic programming.
\end{proof}

Thus, \Cref{thm:sndp} follows.

\subsection{Kernel lower bound}

It is well known that a parameterized problem is fixed-parameter tractable if and only if it admits a kernelization (see~\cite{Cygan:Parameterized:2015}, for example).
By~\Cref{thm:sndp}, \textsc{Shortest Non-Disconnecting Path} admits a kernelization.
A natural step next to this is to explore the existence of polynomial kernelizations for \textsc{Shortest Non-Disconnecting Path}. 
However, the following theorem conditionally rules out the possibility of polynomial kernelization.
To prove this, we first show the following lemma.

\begin{lemma}\label{lem:sndp:cut}
    Let $H$ be a connected graph.
    Suppose that $H$ has a cut vertex $v$.
    Let $C$ be a component in $H - \{v\}$ and let $F \subseteq E(H[C \cup \{v\}])$.
    Then, $H - F$ is connected if and only if $H[C \cup \{v\}] - F$ is connected.
\end{lemma}

\begin{proof}
    If $H - F$ is connected, then all the vertices in $C \cup \{v\}$ are reachable from $v$ in $H - F$ without passing through any vertex in $V(H) \setminus (\{C\} \cup \{v\})$.
    Thus, such vertices are reachable from $v$ in $H[C \cup \{v\}] - F$.
    Conversely, suppose $H[C \cup \{v\}] - F$ is connected.
    Then, every vertex in $C$ is reachable from $v$ in $H - F$.
    Moreover, as $F$ does not contain any edge outside $H[C \cup \{v\}]$, every other vertex is reachable from $v$ in $H - F$ as well.
\end{proof}

\begin{theorem}\label{thm:sndp-kernel}
    Unless \emph{coNP} $\subseteq$ \emph{NP}$/$\emph{poly}, \textsc{Shortest Non-Disconnecting Path} does not admit a polynomial kernelization (with respect to parameter $k$).
\end{theorem}

\begin{proof}
    We give an OR-composition for \textsc{Shortest Non-Disconnecting Path}.
    Let $(G_1, s_1, t_1, k), \ldots, (G_p, s_p, t_p, k)$ be $p$ instances of \textsc{Shortest Non-Disconnecting Path}.
    We assume that for $1 \le i \le p$, $t_i$ is not a cut vertex in $G_i$. 
    To justify this assumption, suppose that $t_i$ is a cut vertex in $G_i$.
    Let $C$ be the component in $G_i - \{t_i\}$ that contains $s_i$.
    By~\cref{lem:sndp:cut}, for any $s_i$-$t_i$ path, it is non-disconnecting in $G_i$ if and only if so is in $G_i[C \cup \{t_i\}]$.
    Thus, by replacing $G_i$ with $G_i[C]$, we can assume that $t_i$ is not a cut vertex in $G_i$.
    
    From the disjoint union of $G_1, \ldots, G_p$, we construct a single instance $(G, s, t, k')$ as follows.
    We first add a vertex $s$ and an edge between $s$ and $s_i$ for each $1 \le i \le p$.
    Then, we identify all $t_i$'s into a single vertex $t$.
    See \Cref{fig:or-comp} for an illustration.
    \begin{figure}
        \centering
        \includegraphics{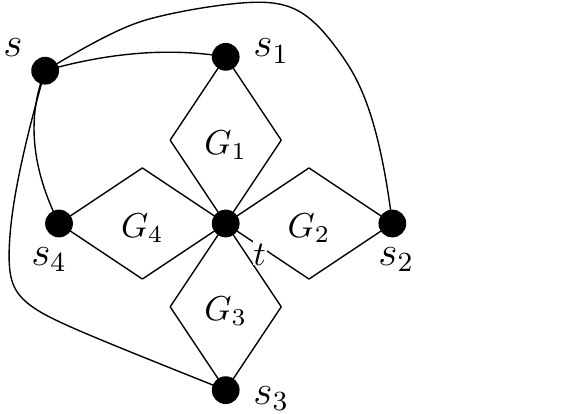}
        \caption{An illustration of the graph $G$ obtained from $q = 4$ instances.}
        \label{fig:or-comp}
    \end{figure}
    In the following, we may not distinguish $t$ from $t_i$.
    Now, we claim that $(G, s, t, k + 1)$ is a yes-instance if and only if $(G_i, s_i, t_i, k)$ is a yes-instance for some $i$.
    
    Consider an arbitrary $s$-$t$ path in $G$.
    Observe that all edges in the path except for that incident to $s$ are contained in a single subgraph $G_i$ for some $1 \le i \le p$ as $\{s, t\}$ separates $V(G_i) \setminus \{t_i\}$ from $V(G_j) \setminus \{t_j\}$ for $j \neq i$.
    Moreover, the path $P$ forms $P = (s, s_i, v_1, \ldots, v_q, t)$, meaning that the subpath $P' = (s_1, v_1, \ldots, v_q, t_i)$ is an $s_i$-$t_i$ path in $G_i$.
    This conversion is reversible: for any $s_i$-$t_i$ path $P'$ in $G_i$, the path obtained from $P'$ by attaching $s$ adjacent to $s_i$ is an $s$-$t$ path in $G$.
    Thus, it suffices to show that for $F \subseteq E(G_i)$, $F \cup \{\{s, s_i\}\}$ is a cut of $G$ if and only if $F$ is a cut of $G_i$.
    Since $t$ is a cut vertex in $G - \{\{s, s_i\}\}$, by~\Cref{lem:sndp:cut}, the claim holds.
\end{proof}

\bibliography{main}

\end{document}